\documentclass[runningheads]{llncs}
\usepackage[in]{fullpage}

\usepackage{booktabs}
\usepackage{bbold}
\usepackage[english]{babel}
\usepackage[utf8]{inputenc}
\usepackage{bm}
\usepackage[linesnumbered,ruled,vlined]{algorithm2e}
\usepackage{xcolor}
\usepackage{mathtools}
\usepackage{multirow}
\usepackage[noend]{algpseudocode}
\usepackage{tabularx}
\usepackage{authblk}
\usepackage{mathtools}
\usepackage[shortlabels]{enumitem}
\usepackage{comment}
\usepackage{caption}
\usepackage{blindtext}

\usepackage{halloweenmath}

\usepackage{setspace}
\usepackage{graphicx}
\graphicspath{ {./images/} }
\usepackage[caption=false]{subfig}

\newcommand{\off}{\textsc{OFF}\xspace}
\newcommand{\oni}{\textsc{ON-I}\xspace}
\newcommand{\onii}{\textsc{ON-II}\xspace}
\newcommand{\on}{\textsc{ON}\xspace}
\newcommand{\coston}{\textsc{Cost}_{\on}}
\newcommand{\costoff}{\textsc{Cost}_{\off}}

\newcommand{\eps}{\varepsilon}
\newcommand{\calO}{\mathcal{O}}
\newcommand{\bal}{\emph{b}}

\newcommand{\slog}{ \lceil \log C \rceil}

\newcommand{\unirej}{unidirectional transaction stream with rejection}
\newcommand{\reject}[1]{R#1 + f_2}
\newcommand{\gor}{\frac{\sqrt{5}-1}{2}}

\newcommand{\nebu}{X}
\newcommand{\rtl}[1]{\protect\overscriptleftarrow{#1}}
\newcommand{\ltr}[1]{\protect\overscriptrightarrow{#1}}

\newcommand{\tracker}{F_{\textit{tracker}}}

\usepackage{blindtext}

\newtheorem{observation}{Observation}

\newcommand{\namedtheoremname}{}
\newtheorem{namedtheoreminner}[theorem]{\protect\namedtheoremname}

\newcommand{\nameddefinitionname}{}
\newtheorem{nameddefinitioninner}[definition]{\protect\nameddefinitionname}

\newcommand{\namedlemmaname}{}
\newtheorem{namedlemmainner}[lemma]{\protect\namedlemmaname}

\newcommand{\namedcorollaryname}{}
\newtheorem{namedcorollaryinner}[corollary]{\protect\namedcorollaryname}

\newcommand{\namedpropositionname}{}
\newtheorem{namedpropositioninner}[proposition]{\protect\namedpropositionname}

\newcommand{\namedexamplename}{}
\newtheorem{namedexampleinner}[example]{\protect\namedexamplename}

\newcommand{\namedexercisename}{}
\newtheorem{namedexerciseinner}[exercise]{\protect\namedexercisename}

\newcommand{\namedobservationname}{}
\newtheorem{namedobservationinner}[observation]{\protect\namedobservationname}

\newcommand{\namedconjecturename}{}
\newtheorem{namedconjectureinner}[conjecture]{\protect\namedconjecturename}

\usepackage{subfiles}
\usepackage{cleveref}
\usepackage{balance}
\usepackage{caption} 
\begin{document}

\title{Online Admission Control and Rebalancing\\in Payment Channel Networks}
\author{Mahsa Bastankhah\inst{1} \and
Krishnendu Chatterjee\inst{2} \and
Mohammad Ali Maddah-Ali\inst{1} \and
Stefan Schmid\inst{3} \and
Jakub Svoboda\inst{2} \and
Michelle Yeo\inst{2}}
\authorrunning{M. Bastankhah et al.}
%
\institute{Sharif University of Technology,\\ \email{mahsa.bastankhah@ee.sharif.edu}, \email{Maddah\_ali@sharif.edu} \and
Institute of Science and Technology Austria, \\
\email{\{Krishnendu.Chatterjee , jakub.svoboda ,  michelle.yeo\}@ist.ac.at} \and
TU Berlin \\\email{stefan.schmid@tu-berlin.de}}
\maketitle

\begin{abstract}
Payment channel networks (PCNs) are a promising technology to improve the scalability of cryptocurrencies. 
PCNs, however, face the challenge that the frequent usage of certain routes may deplete channels in one direction, and hence prevent further transactions. 
In order to reap the full potential of PCNs, recharging and rebalancing mechanisms are required to provision channels, as well as an admission control logic to decide which transactions to reject in case capacity is insufficient. 
This paper presents a formal model of this optimisation problem. 
In particular, we consider an online algorithms perspective, where transactions arrive over time in an unpredictable manner.  
Our main contributions are competitive online algorithms which come with provable guarantees over time. 
We empirically evaluate our algorithms on randomly generated transactions to compare the average performance of our algorithms to our theoretical bounds.
We also show how this model and approach differs from related problems in classic communication networks.
\end{abstract}

\section{Introduction}

Blockchain consensus protocols are notoriously inefficient: for instance, Bitcoin can only support $7$ transactions per second on average which makes it unrealistic to use in everyday situations. 
Payment channel networks like Bitcoin's Lightning Network~\cite{poon2015lightning} and Ethereum's Raiden~\cite{raiden} have been proposed as scalability solutions to blockchains. 
Instead of sending transactions to the blockchain and waiting for the entire blockchain (which can comprise of millions of users) to achieve consensus, any two users that wish to transact with each other can simply open a payment channel between themselves. 
Opening a payment channel requires an initial funding transaction on the blockchain where both users lock some funds only to use in the channel.
Once a payment channel is opened, the channel acts as a local, two-party ledger: payments between the users of channel simply involve decreasing the balance of the payer by the payment amount, and increasing the balance of the payee correspondingly.
 As these local transactions only involve exchanging signatures between the two users and do not involve the blockchain at all, they can be almost instantaneous.
As long as there is sufficient balance, payments can occur indefinitely between two users, until the users decide to close the channel. 
This would involve going back to the blockchain and takes, in the worst case, a small constant number of transactions (in the Lightning Network this constant is 2). 
Thus, with only a small constant number of on-chain transactions, any two users can potentially make arbitrarily many costless transactions between themselves.

Apart from joining a payment channel network to efficiently transact with other users, an additional financial incentive to joining the network is to profit from forwarding transactions. 
Any two users that are not directly connected to each other with a payment channel can transact with each other in a multi-hop fashion as long as they are connected by a path of payment channels. 
To incentivise these intermediary nodes on the path to forward the payment, the network typically allows these nodes to charge a transaction fee. 
Thus, it is common for users to join the network specifically to play the role of an intermediary node that routes transactions, creating channels and fees optimally and selecting the most profitable transactions to maximise their profit from transaction fees \cite{AvarikiotiH0W20,ersoy2019profit}.

However, greedily accepting and routing incoming transactions could rapidly deplete a user's balance in their channels.
In particular, if certain routes are primarily used in one direction, their channels can get depleted, making it impossible to forward further transactions.
Accounting for this problem can be non-trivial, especially since demand patterns are hard to predict and often also confidential. 

To resolve this issue, PCNs typically support two mechanisms:
\begin{itemize}
  \item \emph{On-chain recharging:}  A user can close and reopen a depleted channel with more funds on-chain. 
    \item \emph{Off-chain rebalancing:} An alternative solution is to extend the lifetime of a depleted channel without involving the blockchain,
    by finding a cycle of payment channels in the network to shift funds from one channel to another.
\end{itemize}

Both cases, however, entail a cost.
Intermediaries need to consider the tradeoff between admitting trasactions and potential costs for recharging and rebalancing.
This decision making process is especially important to big routers which are the primary maintainers of payment channel networks like the Lightning Network. 

In this work, we focus on the problem of admission control, recharging and rebalancing in a single payment channel from the perspective of an intermediary node that seeks to route 
as many transactions as possible with minimal costs.
Specifically, we address the following research question:\\

\textit{Can we design efficient online algorithms for deciding when to accept/reject transactions, and when to recharge or rebalance in a single payment channel?}\\

We seek to address this problem with as few restrictions on the actions of the users as possible in order to ensure that our work is realistic. 
Thus, we assume a fixed payment channel network topology with some recharging and rebalancing costs, and a global fee function that is linear in the transaction size. 
We also assume users incur a rejection cost in the form of opportunity cost when they reject to route a transaction.

We are interested in robust solutions which do not depend on any knowledge or assumptions on the demand. Accordingly, we assume that transactions can arrive in an arbitrary order at a channel, and aim to design online algorithms which provide worst-case guarantees. We are in the realm of competitive analysis, and assume that an adversary with knowledge of our algorithms  chooses the most pessimal online transaction sequence. Our objective is to optimise the \emph{competitive ratio}~\cite{borodin2005online}: we compare the 
performance of our online algorithms (to which the transaction sequence is revealed over time) with the optimal offline algorithm that has access 
to the entire transaction sequence in advance.

\subsection{Our contributions}

We initiate the study of a fundamental resource allocation problem in payment channel networks, from an online algorithms perspective.  
Our main result is a competitive online algorithm to admit transaction streams arriving at both sides of a payment channel, and also to recharge and rebalance the channel, in order to maximise the throughput over the channel while accounting for costs. 
In particular, our algorithm achieves a competitive ratio of $7 + 2 \slog$ where $C +1$ is the length of the rebalancing cycle used to replenish the funds on the channel off-chain.
We also provide lower bounds on the amount of funds needed in a channel in order to ensure our algorithm is $c$-competitive for $c < \frac{\log C}{\log \log C}$.

In order to prove our main theorem, we decompose the problem into two simpler but decreasingly restrictive sub problems that may also be of independent interest: 

\begin{enumerate}
  \item \emph{Sub problem 1:}  The first and most restrictive sub problem considers a transaction stream coming only from one direction across a payment channel, and users do not have the option to reject incoming transactions. 
  We present a $2$-competitive algorithm for this problem, which is optimal in the sense that no deterministic online algorithm can achieve a lower competitive ratio. 
  \item \emph{Sub problem 2:}  As a relaxation, our second sub problem allows users to reject transactions although all transactions are still restricted to come from one direction along a payment channel. 
  We show that our algorithm achieves a competitive ratio of $2 + \gor$ for this sub problem. 
  We stress that our lower bound of $2$ we achieve in sub problem $1$ also holds in this sub problem, hence our competitive ratio of $2 + \gor$ is close to optimal.
\end{enumerate}

All these intermediate results as well as the result for our main problem are summarised in \Cref{tab:summary}. 
The algorithms and analysis designed to address these sub problems are eventually used as building blocks for our main algorithm and main theorem.

We complement our theoretical worst-case analysis by performing an empirical evaluation of the performance of our algorithm on randomly generated transaction sequences. 
To this end, we also present an exact dynamic programming algorithm, allowing us to evaluate the competitive ratio. 
We observe that our algorithms perform much better on average compared to our theoretical worst-case bound.

\begin{table*}[h!]
\centering
\begin{tabular}{||c | c||} 
 \hline
 Sub problem & Competitive ratio\\ [0.5ex] 
 \hline\hline
 Unidirectional stream without rejection & $2$  \\ 
 \hline
 Unidirectional stream with rejection & $2 + \frac{\sqrt{5}-1}{2}$   \\
 \hline
 \textbf{Bidirectional stream } & $\bm{7 + 2\slog}$  \\
 \hline
\end{tabular}
\caption{Summary of the theoretical results in our paper. The first column presents each sub problem we analyse in our paper and the second column shows the competitive ratio achieved by our algorithms for each sub problem}
\label{tab:summary}
\end{table*}

\subsection{Related work}

\paragraph{Maintaining balanced payment channels.} As channel balances are typically private, classic transaction routing protocols on payment channel networks like Flare~\cite{Prihodko2016FlareA}, SilentWhispers~\cite{MalavoltaMKM17} and  SpeedyMurmurs~\cite{RoosMKG18} focus mainly on throughput and ignore the issue of balance depletion. Recently, several works shift the focus on maintaining balanced payment channels for as long as possible while ensuring liveness of the network. 
Revive~\cite{KhalilG17} initiated the study rebalancing strategies, Spider\cite{SivaramanVRNYMF20} uses multi-path routing to ensure high transaction throughput while maintaining balanced payment channels, the Merchant\cite{Engelshoven2021TheMA} utilises fee strategies to incentivise the balanced use of payment channels, and \cite{LiMZ20} uses estimated payment demands along channels to plan the amount of funds to inject into a channel during channel creation, to just give a few examples. Our work focuses on minimising costs incurred in the process of handling transactions across a channel and thus we also indirectly seek to maintain balanced payment channels. Moreover, in contrast to previous works which typically assume some form of offline knowledge of the transaction flow in the network, we provide an algorithm which comes with provable worst-case guarantees.

\paragraph{Off-chain rebalancing.} Off-chain rebalancing has been studied as a cheaper alternative to refunding a channel by closing and reopening it on the blockchain. 
In the Lightning Network, there are already several off-chain rebalancing plugins for c-lightning\footnote{https://github.com/lightningd/plugins/tree/master/rebalance} and lnd\footnote{https://github.com/bitromortac/lndmanage}. 
An automated approach to performing off-chain rebalancing using the imbalance measure as a heuristic has been proposed in \cite{PickhardtN20}. Our work similarly studies when to rebalance payment channels, however we make the decision in tandem with other decisions like accepting or rejecting transactions. Recently, \cite{hidenseek} and~\cite{KhalilG17} propose a global approach to off-chain rebalancing where demand for rebalancing cycles is aggregated across the entire network and translated to an LP which is subsequently solved to obtain an optimal rebalancing solution. These approaches are orthogonal and complementary to ours as our focus concerns decision making in a single payment channel and not the entire network.

\paragraph{Online algorithms for payment channel networks.} Online algorithms for payment channel networks have also been studied in \cite{AvarikiotiBWW2019} and \cite{FazliNS2021}. Avarikioti et al. \cite{AvarikiotiBWW2019} establish impossibility results against certain classes of adversaries, however they only consider a limited problem setting where their algorithms can only accept or reject transactions (with constant rejection cost). Fazli et al. \cite{FazliNS2021} considers the problem of optimally scheduling on-chain recharging given a sequence of transactions. In contrast to previous work, our work considers a more general problem setting where our algorithms can not only accept or reject transactions, but also recharge and rebalance channels off chain. We also extend the cost of rejection to take into account the size of the transaction. 

\paragraph{Relationship to classic communication networks.} Admission control problems such as online call admission \cite{aspnes1997line,lukovszki2015online} are fundamental and have also received much attention in the context of communication networks. 
However, in classic communication networks the available capacity of a link in one direction is independent of the flows travelling in the other direction, and moreover, link capacities are only consumed by the currently allocated flows. 
In contrast, the capacities of links in payment-channel networks are permanently reduced by transactions flowing in one direction, but can be topped up by flows travelling in the other direction. 
The resulting rebalancing opportunity renders the underlying algorithmic problem significantly different.

\subsection{Organisation}
The remainder of this paper is organised as follows. In \Cref{sec:model}, we introduce our model and give a brief overview of the various sub problems we consider. \Cref{sec:tracker} presents some algorithmic building blocks. We describe and analyse our algorithms for the sub problems in \Cref{sec:accept-all} and \Cref{sec:rejection}, and for the general problem in \Cref{sec:two_parties}. We empirically evaluate our online algorithms in \Cref{sec:eval}, before concluding our contribution in \Cref{sec:conclusion}.

\section{Model}\label{sec:model}

We introduce preliminaries, and present our model (and its sub problems).

\subsection{Modelling PCNs}

\paragraph{Payment channels.}
We model the payment channel network as an undirected graph $G=(V,E)$.
A payment channel between users $\ell,r \in V$ is an edge $(\ell,r) \in E$.
By denoting parties in a given channel as $\ell$ (left) and $r$ (right), we benefit from a visual representation of a channel in which transactions move funds across the channel from left to right or from right to left. We denote the balance of user $\ell$ (resp. $r$) in the channel $(\ell,r)$ by $\bal(\ell)$ (resp. $\bal(r)$). The \emph{capacity} of the channel is the total amount of funds locked in the channel. That is, for a channel $(\ell,r)$, the capacity of $(\ell,r)$ is $\bal(\ell) + \bal(r)$. A left-to-right transaction of amount $x$ decreases $\ell$'s balance by $x$ and increases $r$'s balance by $x$ and vice versa for a right-to-left transaction of $x$.

\paragraph{Recharging and rebalancing payment channels.}
When a user in a channel does not have sufficient funds to accept a transaction, the user can either reject the transaction, recharge the channel, or rebalance the channel. 
Recharging the channel happens on-chain and corresponds to closing the payment channel on the blockchain and opening a new channel with more funds. Refer to \Cref{fig:recharging} for an example. 
In contrast, rebalancing the channel happens entirely off-chain.
Here, users find a cycle of payment channels to shift funds from one of their other channels to refund the depleted channel. 
\Cref{fig:off-chain-reb} depicts a simple example of off-chain rebalancing where user $\ell$ wishes to refund their depleted $(\ell,r)$ channel with amount $5$.

\paragraph{Transactions.} We consider a transaction sequence $X_t = (x_1, ..., x_t)$, $x_i \in \mathbb{R}^+$, that arrives at a fixed payment channel online. 
Transactions are processed in the order in which they arrive in the sequence.
Each transaction $x_i$ in the sequence has both a value and a direction along a payment channel. 
The value of a transaction is simply the amount that is being transferred.
The direction of a transaction across a payment channel $(\ell, r)$ determines who is the sender of the transaction and who is the receiver. 
When we have a sequence of transactions that go in both directions along a payment channel, we use $\ltr{x}$ to denote a transaction that goes from left-to-right and $\rtl{x}$ to denote a transaction that goes from right-to-left.
We say a user, wlog $\ell$, \emph{accepts} a transaction of size $x$ coming from the left to right direction along the channel $(\ell,r)$ if $\ell$ agrees to forward $x$ to $r$ (see \Cref{fig:accept} for an example of $r$ accepting transaction $\rtl{1}$). 
Similarly, we say a user $\ell$ \emph{rejects} a transaction $x$ coming from the left to right direction along the channel $(\ell,r)$ when $\ell$ does not forward the transaction to $r$ (see \Cref{fig:reject} for an example of $r$ rejecting transaction $\rtl{9}$). 
When it is clear which channel and direction we are referring to, we simply say $\ell$ accepts or rejects $x$. 

\paragraph{Costs.}
We consider three types of costs in our problem setting:
\begin{enumerate}
    \item \textbf{Rejecting transactions:} For any fixed user $\ell$, the revenue in terms of transaction fees from forwarding a payment of size $x$ is $Rx + f_2$, where $R, f_2 \in \mathbb{R}^+$. Consequently, the cost of rejecting a transaction of size $x$ is simply the opportunity cost of gaining revenue from accepting the transaction, i.e. $Rx + f_2$.
    \item \textbf{On-chain recharging:} For any user $\ell$, the cost of recharging a channel on-chain is $F + f_1$, where $F$ is the amount of funds $\ell$ puts into the new channel
    and $f_1 \in \mathbb{R}^+$ is an auxiliary cost independent of $F$ which captures the on-chain recharging transaction fee. 
    \item \textbf{Off-chain rebalancing:} For any user $\ell$, the cost of off-chain rebalancing for an amount $x$ is $C\cdot(Rx + f_2)$, where $C$ is the length of the cycle along which funds are sent $-1$. In the example of off-chain rebalancing in \Cref{fig:off-chain-reb}, the length of the rebalancing cycle is 3 and thus $C=2$. 
    
\end{enumerate}

Let us denote by $\off$ the optimal offline algorithm and $\on$ an online deterministic algorithm. 
We denote by $\coston(X_t)$ (resp. $\costoff(X_t)$) the total cost of $\on$ (resp. $\off$) given the transaction sequence $X_t$.

\begin{figure}[h]
    \centering
    \includegraphics[scale=0.7]{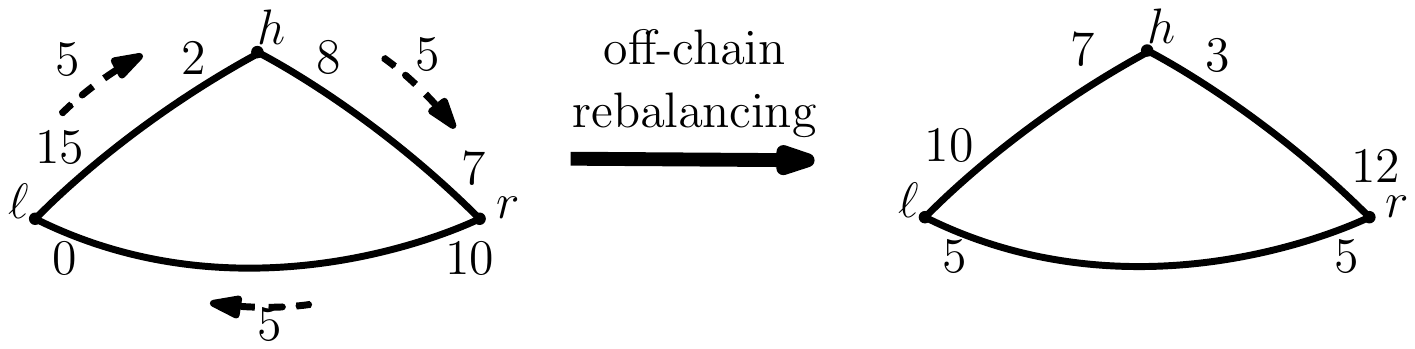}
    \caption{Example of off-chain rebalancing in a payment channel network with users $\ell$, $h$, and $r$.  
    $\ell$ transfers $5$ coins along the channels $(\ell, h)$, $(h, r)$ and $(r, \ell)$ to replenish $\ell$'s balance on the $(\ell, r)$ channel.
    The graph on the right depicts the balances of each user in each channel after performing off-chain rebalancing.}
    \label{fig:off-chain-reb}
\end{figure}

\begin{figure}
\centering

    \subfloat[Accepting a transaction]{{\includegraphics[width=0.45\columnwidth]{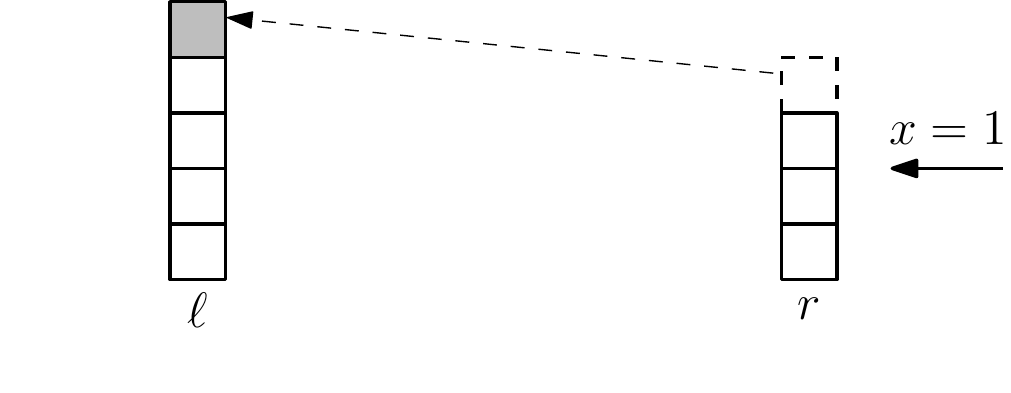}\label{fig:accept} }}%
    \qquad\qquad
    \subfloat[Rejecting a transaction]{{\includegraphics[width=0.45\columnwidth]{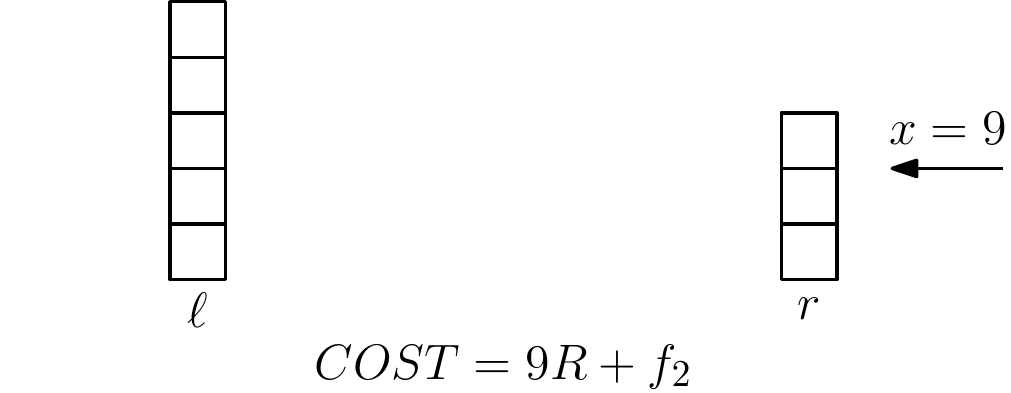}\label{fig:reject} }}%
    
     \subfloat[\centering Off-chain rebalancing]{{\includegraphics[width=0.45\columnwidth]{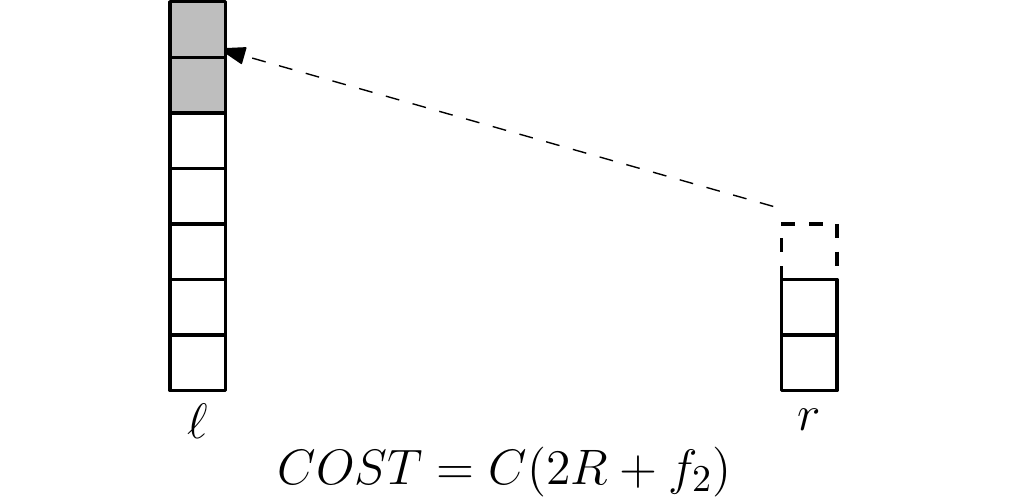}\label{fig:reb} }}%
    \qquad\qquad
    \subfloat[\centering Recharging]{{\includegraphics[width=0.45\columnwidth]{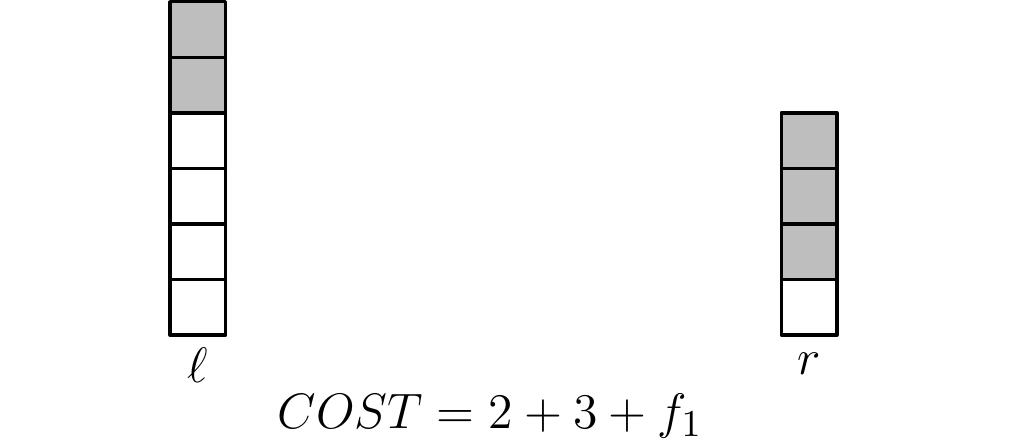}\label{fig:recharging} }}%




\caption{Example of actions users $\ell$ and $r$ can take in the general bidirectional stream setting, and the changes in the balances of both users as a result of these actions. Each square represents $1$ coin. }
  \label{fig:actions}
 \end{figure}

\paragraph{Competitive ratio.}
We say an online algorithm $\on$ is \emph{c-competitive} if for every transaction sequence $X_t$ generated by the adversary,
$$\coston(X_t) \leq c \cdot \costoff(X_t)$$

\subsection{Main and sub problems}\label{sec:problems}
Our main problem is to design a competitive deterministic online algorithm that determines when to accept/reject transactions and when to recharge or rebalance the channel given a bidirectional stream of transactions across a payment channel. 
More precisely, we consider a stream of transactions that can arrive from both right to left or left to right in a given payment channel $(\ell, r)$. 
$\ell$ (resp. $r$) can choose to accept or reject transactions coming in the left-to-right (resp. right-to-left) direction in the stream. 
Either user would incur a cost of $Rx + f_2$ for rejecting a transaction of size $x$.
Both users can also recharge the channel on-chain at any point, incurring a cost of $F_{\ell} + F_r + f_1$ where $F_{\ell}$ and $F_r$ are the funds put into the channel by $\ell$ and $r$ respectively.
Since transactions are streaming in both directions in this model, both users would incur costs in this setting.
Thus, we seek to design an algorithm that minimises the cost of the \emph{entire} channel. 
Refer to \Cref{fig:actions} for examples of the actions that a user can take in our main bidirectional transaction stream setting. 

To this end, we give a formal definition of two sub problems of decreasing restrictiveness on the part of user actions. 
We present these sub problems as the algorithms and analysis used to solve these sub problems are used in developing the algorithm and analysis for our main problem.

\paragraph{Unidirectional stream without rejection}
In this model, we consider the case where transactions stream only in one direction along a given payment channel. 
Here, users cannot reject incoming transactions, or that doing so is not worth it (e.g., when the cost of rejecting a transaction is larger than recharging the channel). 
Formally, given a channel $(\ell,r)$ and a transaction stream from wlog left to right, user $\ell$ only accepts a transaction $x$ if $\bal(\ell) >x$.
Otherwise, $\ell$ has to recharge the channel on-chain with more funds, incurring a cost of $F + f_1$ where $F$ is the amount of funds $\ell$ adds to the channel.
As we only consider transactions streaming in one direction, only one user would incur costs in this setting (the user that has to decide whether to accept or reject transactions). 
A real world example that motivates this setting is a company which wants to position itself as a ``routing hub" in a payment channel network, providing a routing service in return for transaction fees. 
As such, the company would want to accept as many transactions possible to acquire the reputation of a hub that is constantly available.

\paragraph{Unidirectional stream with rejection}
In this model, we still restrict the transaction stream from wlog left to right in a given payment channel $(\ell, r)$. 
However, $\ell$ can reject transactions, incurring a rejection cost of $Rx + f_2$ for a transaction of size $x$. 
$\ell$ can also recharge the channel on-chain at any point, incurring a cost of $F + f_1$ where $F$ is the amount of funds $\ell$ adds to the channel.

\section{Algorithmic Building Blocks}\label{sec:tracker}

Before we describe and analyse the performance of our algorithms in the various settings outlined in \Cref{sec:problems},
we introduce two algorithmic building blocks that we use extensively in the later sections of our work.
The first building block is an algorithm \Call{Funds}{}.
It takes a sequence of transactions as an input and returns the amount of funds that an optimal algorithm uses on this sequence.
The purpose of the algorithm is to track the funds \off has in their channel assuming that the sequence of transactions ends at this point.
For first two sub problems we show how to compute \Call{Funds}{}.
For the main problem, we propose a dynamic programming approach in \Cref{app:dp}.
The second building block is a general recharging online algorithm that calls \Call{Funds}{} as a subroutine and uses the output to decide when and how much to recharge the channel on-chain.
The general idea behind the recharging online algorithm is to recharge whenever the amount of funds in \off's channel ``catches up" to the amount of funds \on has in their channel.

\subsection{Tracking funds of \off}

For a given transaction sequence $X_t = (x_1, \dots, x_t)$, let us denote $A(X_i)$ to be the amount of funds \off would use in the channel if \off gets the sequence $X_i = (x_1, \dots, x_i)$ (i.e. the length $i$ prefix of $X_t$) as input. 
By appending subsequent transactions $x_{i+1}, \dots, x_{t}$ from $X_t$ to $X_i$, we can view $A(X_i)$ as a partial solution to the online optimisation problem that gets updated whenever a new transaction arrives online. 
In the unidirectional transaction stream (with or without rejection) setting, $A(X_i)$ refers to all the funds a fixed user locks into a payment channel.
In the bidirectional transaction stream setting, $A(X_i)$ refers to the total balance of both users in the channel.
We assume that given an input sequence $X_t$, $\Call{Funds}{X_t}$ performs the necessary computations and returns $A(X_t)$.
For our main problem, computing $\Call{Funds}{X_t}$ is generally NP-hard, but we can approximate it to a constant factor, see~\cite{oracle} for more details.

\subsection{Using tracking for recharging}\label{subsec:tracking-alg}

In \Cref{Algorithm:onlinetracker}, we describe an online $(\gamma, \delta)$-recharging algorithm \on that uses \Call{Funds}{} as a subroutine to decide when and how much to recharge the channel.
\on is run by one user (wlog $\ell$) in a payment channel $(\ell, r)$.
\on calls \Call{Funds}{}  
after each transaction to check if the new transaction sequence results in a significant increase in the amount of funds \off has in their channel.
Whenever \on notices that \off's funds have increased above the threshold (\Cref{A2condition} in \Cref{Algorithm:onlinetracker}), \on recharges the channel with an amount of $\gamma (A(X_i) + \delta)$ where $A(X_i)$ is the amount of funds \off has in their channel.
\Cref{tab:examplesequence} depicts an example sequence of the amount of funds inside the channel of \off, the value of the tracking variable, and the amount of funds inside the channel of \on. 

Let us denote $A_t \vcentcolon = \max_{i \le t} A(X_i)$. Now we state and prove some important properties of the $(\gamma, \delta)$-recharging algorithm.

\RestyleAlgo{ruled}

\SetKwFunction{FTracker}{Tracker}

\SetKwComment{Comment}{ }{}
\SetKwInput{KwInitialise}{Initialise}
\SetKwFunction{FTracker}{Tracker}
\begin{algorithm}[t]
  \caption{$(\gamma, \delta)$-recharging}
  \label{Algorithm:onlinetracker}
  \KwInitialise{ $\tracker, X \gets 0, \emptyset$}
  \For{transaction $x$ in order of arrival}{
    concatenate $x$ to $X$\\
    $\tracker' \gets \Call{Funds}{X}$\\
    \uIf{$\tracker' > \tracker$}{ \label{A2condition}
        $\tracker \gets \tracker' + \delta$\\
        recharge to $\gamma \tracker$
    }
  }
\end{algorithm}

\begin{lemma}\label{Lemma: Tracking}
  \Cref{Algorithm:onlinetracker}  with parameters $(\gamma, \delta)$ ensures that \on always has at least $\gamma$ times the amount of funds \off has and ensures that \on incurs a cost of at most $\gamma(A_t + \delta) + f_1 \cdot \lceil \frac{A_t}{\delta} \rceil$.
\end{lemma}
\begin{proof}
The first part of the claim follows from the fact that the moment $A(X_i) > \tracker$ for some $i$, $\tracker$ gets updated to $A(X_i) + \delta > A(X_i)$ and \on recharges the channel to $\gamma \tracker > \gamma A(X_i)$.

For the second part of the claim, we note that the cost incurred by \on is simply the total amount of funds added to the channel with an additional cost of $f_1$ each time \on recharges the channel on-chain.
  The amount of funds locked in the channel for \on is always at most $A_t + \delta$ and the times when \on recharges the channel occurs whenever \off increases its funds by an amount of at least $\delta$.
  Thus, the number of rechargings for \on that can occur is at most $\lceil \frac{A_t}{\delta} \rceil$ with a cost of $f_1$ for each recharging instance.
  The total cost incurred by \on is therefore $\gamma(A_t + \delta) + f_1 \cdot \lceil \frac{A_t}{\delta} \rceil$.
\end{proof}

Next, we show a simple lower bound in terms of $A_t$ for the cost of \off given a sequence of transactions $X_t$.

\begin{lemma}\label{Lemma: Money cost}
  If $A_t > 0$, then $\costoff(X_t)$ is at least $A_t + f_1$.
\end{lemma}
\begin{proof}
We first note that the sequence of costs for \off is monotonically increasing, i.e. $\costoff(X_t) \le \costoff(X_{t+1})$.
  This comes from the fact that any action of \off at step $i$ of the sequence can only increase its cost (i.e. either rejecting $x_{i+1}$ or recharging the channel and then accepting $x_{i+1}$), or it does not change the cost at all (i.e. by accepting $x_{i+1}$ without recharging).

Since $A_t >0$, we know that \off recharged on-chain at some point to an amount $A_t$ for a recharging cost of $A_t+f_1$. Since the sequence of costs for \off is monotonically increasing, $\costoff(X_t) \ge A_t + f_1$.
\end{proof}

\begin{table}[htb!]
\centering
\begin{tabular}{||c c c||} 
 \hline
 $A(X_i)$ & $\tracker$ & Amount locked in the channel for \on \\ [0.5ex] 
 \hline\hline
 $0$ & $0$ & $0$ \\ 
 \hline
 $\eps$ & $\delta + \eps$ & $\gamma (\delta + \eps)$  \\
 \hline
 $\delta$ & $\delta + \eps$ & $\gamma (\delta + \eps)$ \\
 \hline
 $\delta+2\eps$ & $2(\delta + \eps)$ & $2\gamma (\delta + \eps)$ \\
 \hline
\end{tabular}
\caption{\on is the $(\gamma, \delta)$- recharging algorithm as described in \Cref{Algorithm:onlinetracker}.
  An example sequence of the amount locked in the channel of \off, the value of the tracker $\tracker$, and the amount locked in the channel of \on after 4 transactions.
  We assume $\delta > \eps >0$.}
\label{tab:examplesequence}
\end{table}

\section{Step 1: Unidirectional transaction stream without rejection}\label{sec:accept-all}

In this section we consider the first sub problem where, given a payment channel $(\ell, r)$, transactions stream along the channel in only one direction (wlog left to right). 
Moreover, $\ell$ has to accept an incoming transaction of size $x$ and forward it to $r$ if $\ell$'s balance $\bal(\ell) \ge x$.
Otherwise, $\ell$ needs to recharge the channel on-chain (and accept the transaction after). 

The optimal offline algorithm \off follows a simple strategy:
since it knows the entire stream of transactions in advance, it makes a single recharging action at the beginning of the transaction sequence $X_t$ of size $\sum_{i=1}^t x_i$.
The cost incurred by \off is thus $f_1 + \sum_{i=1}^t x_i$. 

\RestyleAlgo{ruled}
\SetKwComment{Comment}{/* }{ */}
\SetKwInput{KwInitialise}{Initialise}
\SetKwFunction{FTracker}{Tracker}
\begin{algorithm}[t]
  \caption{Unidirectional transaction stream without rejection}
  \label{Algorithm:acceptall}
  \KwInitialise{tracker $\tracker, X \gets 0, \emptyset$}
  \KwInitialise{balance $b= 0$}
  \For{transaction $x$ in order of arrival}{
    concatenate $x$ to $X$\\
    $\tracker' \gets \Call{Funds}{X}$\\
    \uIf{$\tracker' > \tracker$}{
        $\tracker \gets \tracker' + f_1$\\
        recharge to $\tracker$
    }
    Accept $x$
  }
\end{algorithm}

Now, we present a $2$-competitive online algorithm \on for this sub problem (\Cref{Algorithm:acceptall}).
\on uses $(\gamma, \delta)$-recharging with parameters $\gamma=1$ and $\delta=f_1$.
The recharging ensures that \on always has enough funds to accept a transaction.
The following theorem shows that \on is $2$-competitive.

\begin{theorem}\label{theorem_accept_2_comp}
  \Cref{Algorithm:acceptall} is $2$-competitive in the unidirectional transaction stream without rejection.
\end{theorem}
\begin{proof}
  From \Cref{Lemma: Tracking}, setting $\gamma =1$ and $\delta =f_1$ gives \on a cost of at most $A_t + f_1 + f_1 \cdot \lceil \frac{A_t}{f_1} \rceil$.
  Since $f_1\cdot \lceil \frac{A_t}{f_1} \rceil \le f_1 \cdot (\frac{A_t}{f_1} + 1) = A_t + f_1$, the cost of \on is at most $2(A_t + f_1)$.
  From \Cref{Lemma: Money cost}, we know that the cost of \off is at least $A_t + f_1$. Thus, \on is $2$-competitive.
\end{proof}

In addition, we note that \on is optimal in this setting.
The next theorem proves that no deterministic algorithm can achieve a strictly smaller competitive ratio compared to \on.
The proof shows that \on cannot lock too much funds into the channel, otherwise \on's cost is too high,
but if \on locks too little funds, it needs to recharge often.

\begin{theorem}\label{thm:acceptallopt}
  There is no deterministic algorithm that is $c$-competitive for $c<2$ in the unidirectional transaction stream without rejection sub problem.
\end{theorem}
\begin{proof}
  We prove the theorem by contradiction.
  For the sake of contradiction, suppose that there exists a $c$-competitive algorithm \on for $c = 2-\eps$ for some $\eps >0$.
  Consider the following sequence of transactions: $\frac{\eps}{3}, \frac{\eps}{3}, \frac{\eps}{3},  \dots$.
  We note that when the sequence of transactions is of length $k$, the cost of \off is $f_1 + k\cdot \frac{\eps}{3}$ as the optimal solution is to recharge the channel at the start of the sequence to the total sum of the transactions in the sequence.

  For \on to remain $(2-\eps)$-competitive after processing the first transaction, \on locked at most $f_1 - \eps f_1 + \frac{\eps}{3}$ in the channel ($\coston(X_1) = 2f_1 -\eps f_1 +2\frac{\eps}{3}$).
  
  We generalize the above idea and show that either \on has always smaller amount than $f_1 - \frac{\eps}{3}$ in the channel or at some point it has at least $f_1- \frac{\eps}{3}$.
  In both cases, we derive a contradiction to the $(2-\eps)$ competitive ratio of \on.

  First, suppose that \on always recharges to at most $f_1 - \eps'$ for some $\eps' > 0$.
  Then after $t$ transactions, the number of rechargings is at least $\lceil \frac{t \frac{\eps}{3}}{f_1 - \eps'} \rceil$.
  So $\coston(X_t) \ge t \cdot \frac{\eps}{3} + f_1\lceil \frac{t \frac{\eps}{3}}{f_1 - \eps'} \rceil$.
  Setting $t = \frac{3k(f_1 - \eps')}{\eps}$ for some $k$ gives $\costoff(X_t) = f_1 + k(f_1 - \eps')$ and $\coston(X_t) = k(f_1 - \eps') + kf_1$, but since
  $\lim_{k \to \infty} \frac{2kf_1 - k \eps'}{(k+1)f_1 - k\eps'} = \frac{2f_1 - \eps'}{f_1 - \eps'}$ for any $\eps'$, then the competitive ratio is at least $2$.

  Now, suppose that $t$ is the first time that after processing a transaction, \on has at least $f_1 - \frac{\eps}{3}$ locked in the channel.
  At time $t$, $\costoff(X_t) = f_1 + t\frac{\eps}{3}$. Cost of \on is $f_1 + t \frac{\eps}{3}$ for funds locked in the channel plus any additional recharging cost.
  But since it is the first time \on recharged by more than $f_1$, the cost for recharging is $f_1\lceil \frac{t\frac{\eps}{3}}{f_1-\eps'} \rceil \ge f_1 + \frac{t\eps}{3}$ for some other positive $\eps'$.
  So again, $\coston(X_t) \ge t \frac{\eps}{3} + f_1 + t \frac{\eps}{3} + f_1 = 2(f_1 + t\frac{\eps}{3})$ which is twice of $\costoff(X_t)$.
  
  In both cases the cost of \on is at least twice that of \off which contradicts the assumption that \on is $(2-\eps)$-competitive.
\end{proof}

\section{Step 2: Unidirectional transaction stream with rejection}\label{sec:rejection}

In this section we consider the second sub problem where transactions are still streaming along a given payment channel $(\ell, r)$ in one direction (wlog left to right).
This time though, a user can choose to reject incoming transactions.
We describe an algorithm (detailed in \Cref{Algorithm:reject}) with competitive ratio $2+\gor$.
We note that the competitive ratio for this setting is larger than the competitive ratio we achieve in the previous setting as \off has a wider range of decisions.

Let us call a transaction of size $x$ \emph{big} if $x > \reject{x}$ and \emph{small} otherwise.
We first observe that \off in this setting always rejects big transactions. 

\begin{lemma}\label{lemma:rejectbig}
  \off rejects all big transactions in the \unirej.
\end{lemma}
\begin{proof}
  Accepting a transaction $x$ incurs a cost of $x$ for increasing funds.
  Rejecting a transaction $x$ incurs a cost of $\reject{x}$.
  So any big transaction should be rejected.
\end{proof}

Thus, the strategy of \off in this setting is to simply reject all big transactions.
Moreover, if there are sufficiently many small transactions in the sequence to offset the cost of recharging, \off makes a single recharging action at the beginning of the sequence of size $\sum_{x \in X_t, x \text{ is  small }} x$ for a cost of $f_1 + \sum_{x \in X_t, x \text{ is small }} x$. 

\RestyleAlgo{ruled}

\SetKwComment{Comment}{/* }{ */}
\SetKwFunction{FTracker}{Tracker}
\SetKwInput{KwInitialise}{Initialise}
\begin{algorithm}[t]
  \caption{Unidirectional transaction stream with rejection}
  \label{Algorithm:reject}
  \KwInitialise{tracker $\tracker, X \gets 0, \emptyset$}
  \KwInitialise{balance $b=0$}
  \For{transaction $x$ in order of arrival}{
    concatenate $x$ to $X$\\
    $\tracker' \gets \Call{Funds}{X} $\\
    \uIf{$\tracker' > \tracker$}{
        $\tracker \gets \tracker' + \gor f_1$\\
        recharge to $\tracker$
    }
    \uIf{ $b \ge x$ and $x$ is small }{
      Accept $x$ \label{A4accept}
    }
    \uElse{
    Reject $x$
    }
  }
\end{algorithm}

\Cref{Algorithm:reject} performs $(1, \gor f_1)$-recharging and it accepts a transaction $x$ if it has enough funds and $x$ is small.
The following theorem states that \on is $(2+\gor)$-competitive in this problem setting.

\begin{theorem}\label{Theorem: rejection}
  \Cref{Algorithm:reject} is $(2+\gor)$-competitive in the \unirej sub problem.
\end{theorem}
\begin{proof}
  From \Cref{lemma:rejectbig}, \off rejects big transactions. Thus, \on should also reject these transactions.

  While $A_t = 0$, both \off and \on reject all transactions in the sequence and both incur the same cost.
  The moment $A_t > 0$, we know that \off recharged with an amount at least $A_t$ to accept all small transactions in the sequence.
  Thus $\costoff(X_t) \ge A_t + f_1$.

  At this time \on would have rejected the small transactions in the sequence for at most a cost of $A_t + f_1$ together with some additional recharging cost.
  From \Cref{Lemma: Tracking}, we know that the recharging cost for \on is at most
  \[
      A_t + \gor f_1 + \left\lceil \frac{A_t}{\gor f_1} \right\rceil f_1
      \le A_t + \gor f_1 + \frac{A_t}{\frac{\sqrt{5}-1}{2}} +f_1.
  \]
  Summing up both costs, we get 
  \begin{align*}
      \coston(X_t) \le& A_t + f_1 + A_t + \frac{\sqrt{5} -1}{2}f_1 + \frac{A_t}{\frac{\sqrt{5}-1}{2}} +f_1\\
      =& \left(2 + \frac{\sqrt{5}-1}{2}\right) \costoff(X_t)
  \end{align*}
\end{proof}

Before analysing the optimality of \on, we first observe, as a simple corollary of \Cref{thm:acceptallopt}, that the lower bound of $2$ also holds for this sub problem.

\begin{corollary}
There is no deterministic algorithm that is $c$-competitive for $c<2$ in the unidirectional transaction stream with rejection sub problem.
\end{corollary}

We conjecture that no other deterministic algorithm can perform better that \on in this setting.
Moreover, we sketch an approach to prove the conjecture in \Cref{app:rejection_sketch}.

\begin{conjecture}\label[conjecture]{conjecture:rejection}
  There is no deterministic algorithm that is $c$-competitive for $c<2 + \gor$ in the unidirectional transaction stream with rejection setting.
\end{conjecture}

\section{Final step: Bidirectional transaction stream}\label{sec:two_parties}
In this section, we consider the most general problem setting, where for a given payment channel $(\ell, r)$, transactions stream along the channel $(\ell, r)$ in both directions.
A user $\ell$ (resp. $r$) can accept or reject incoming transactions that stream from left to right (resp. right to left).
Either user would incur a cost of $Rx + f_2$ for rejecting a transaction of size $x$.
$\ell$ does not need to take any action when encountering transactions that stream from right to left as they simply increase the balance of $\ell$ in the channel $(\ell, r)$.
Both users can also decide at any point to recharge their channel on-chain, or rebalance their channel off-chain. 

\subsection{Main algorithm}
Here we present our online algorithm for the bidirectional transaction stream setting. 
For simplicity, we assume that $R=0$ in the rejection cost. 
This means that the cost of rejecting a single transaction of size $x$ is simply $f_2$. 
Likewise, since we assume $R=0$, rebalancing an amount of $x$ off-chain now only incurs a cost of $C f_2$. 

The main algorithm \on is detailed in \Cref{Algorithm:logC-competitive-main}. 
It is run by both users on a payment channel.
In a nutshell, \on is composed of three algorithms: the first is a recharging algorithm to determine when and how much to recharge the channel on-chain.
The second algorithm (\Cref{Algorithm:decide_on_trx}) decides whether to accept or reject new transactions and when to perform off-chain rebalancing.
The last algorithm (\Cref{Algorithm:handle-fund}) describes how to store the funds received from the other user of the channel.

\begin{figure}
\centering
  \includegraphics[scale=0.5]{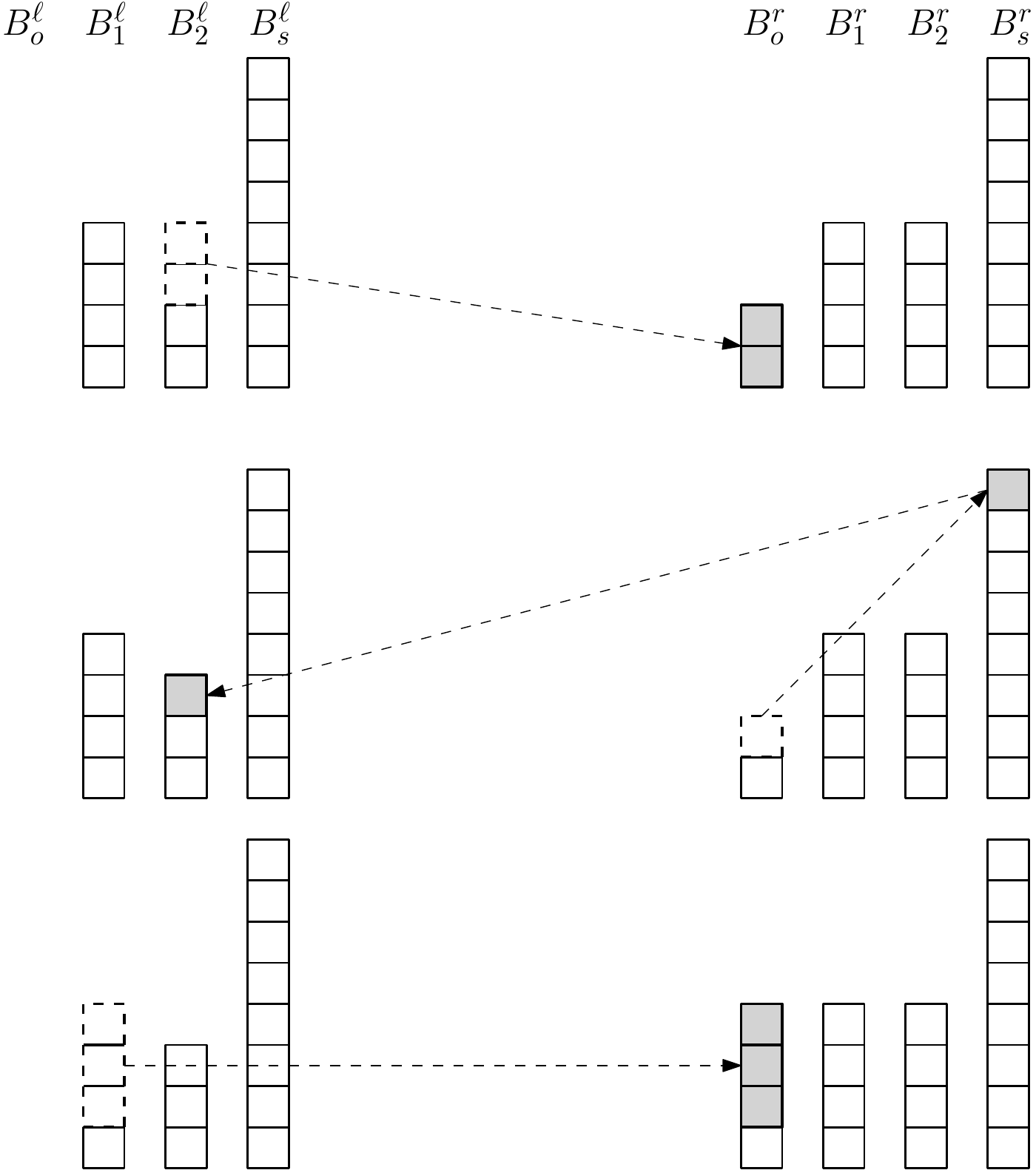}
  \caption{An example of how funds are transferred across a payment channel and how buckets are refilled.
  Both $\ell$ and $r$ start with full buckets.
  The first transaction is in the left-to-right direction and is transferred using funds from $B_2^{\ell}$ to $B_o^r$. The second transaction is in the right-to-left direction and is small, thus funds from $B^r_s$ are used.
  $B^r_s$ is immediately refilled using funds from $B^r_o$.
  The third transaction is in the left-to-right direction and uses funds from $B_1^{\ell}$.}
  \label{fig:bucket}
 \end{figure}
 
\subsubsection{$(4+ 2\slog, f_1)$-recharging}\label{subsub:rebalancing-tracking}
\on runs an on-chain recharging algorithm similar to \Cref{Algorithm:onlinetracker} (see \Cref{mainalg:oracle} and \Cref{mainalg:recharge} in \Cref{Algorithm:logC-competitive-main}) but with parameters $\gamma = 4+2\slog$ and $\delta = f_1$.
Since we are in the bidirectional transaction stream setting, \Call{Funds}{} returns the amount of funds \off has inside the entire channel (i.e. $\bal(\ell) +\bal(r)$) given a transaction sequence.

Let us look at the period between the on-chain recharging instances of \on.
From \Cref{mainalg:recharge} in \Cref{Algorithm:logC-competitive-main}, we know that \on ensures that it has more than $4+2\slog$ times more funds than \off locked in the channel.
These funds are distributed in the following way: \on initialises $\slog + 2 $ ``buckets" on each end of the channel.
We denote set of left-side buckets as $B^{\ell}$ and it consists of $B_s^{\ell} , B_1^{\ell}, \dots, B_{\slog}^{\ell}, B_o^{\ell}$.
Likewise, the set of right-side buckets is $B^r$ and it consists of $B_s^{r} , B_1^{r}, \dots, B_{\slog}^{r}, B_o^{r}$.

After recharging, users decide how to distribute funds in the channel, so the buckets $B_s^{\ell}$ and $B_s^r$ are filled with $2\tracker$ funds.
Buckets $B_o^{\ell}$ and $B_o^r$ are empty ($0$ funds).
Other buckets contain $\tracker$ funds.

Looking ahead, the funds in the $i$-th bucket on both sides are used to accept transactions $x$ with a size in the interval $\left[\frac{\tracker}{2^{i}}, \frac{\tracker}{2^{i-1}}\right)$.
The funds in $B_s$ are used to accept transactions with a size less than $\frac{\tracker}{C}$.
Finally, $B_o$ stores excess funds coming from payments from the other side when all other buckets are full.

\subsubsection{Transaction handling}
When a transaction arrives at the channel, based on the direction of the transaction, either $\ell$ or $r$ executes \Cref{Algorithm:decide_on_trx} to decide whether to accept the transaction.
Wlog let us assume $\ell$ encounters transaction $\ltr{x}$. 
If $\frac{\tracker}{2^i} < x \leq \frac{\tracker}{2^{i-1}}$ for some $i \in [\slog]$ and $B_i^{\ell}$ has sufficient funds, the funds from $B_i^{\ell}$ are used to accept the transaction.
If $B_i^{\ell}$ lacks sufficient funds for accepting $x$, $\ell$ rejects $x$.

Now, we consider the case where $x \le \frac{\tracker}{C}$.
If $B_s^{\ell}$ has sufficient funds, $\ell$ uses the funds from $B_s^{\ell}$ to accept $x$.
If $B_s^{\ell}$ has insufficient funds to accept $x$, $\ell$ performs off-chain rebalancing with an amount such that after deducting $x$ from $B_s^{\ell}$, there would still be $2\tracker$ funds left in $B_s^{\ell}$.
$\ell$ subsequently accepts $x$. 
The required funds for off-chain rebalancing are transferred from $B_o^{r}$ and $B_s^{r}$ (see \Cref{decidealg:bo} and \Cref{decidealg:bs} in \Cref{Algorithm:decide_on_trx}). 
Whenever $B_o^{\ell} > 0$ and some bucket in $B^{\ell}$ gets under its original capacity, funds are reallocated from $B_o^{\ell}$ to fill the bucket. 
\Cref{fig:bucket} depicts an example of how funds are used from different buckets to accept transactions.

\subsubsection{Handling funds coming from the other side}\label{subsubsec:handleComingTrx}
When a transaction $x$ is accepted by wlog $\ell$, \on calls \Cref{Algorithm:handle-fund} to distribute the transferred funds among $r$'s buckets in the following way: $r$ first uses $x$ to fill $B_s^r$ up to its capacity of $2\tracker$ (see \Cref{handlealg:bs} in \Cref{Algorithm:handle-fund}).
If there are still funds left, $r$ refills the $B_i^r$ buckets in descending order from $i=\slog$ to $i=1$. 
Intuitively, the reason why buckets are refilled in descending order is due to our simplified cost model for this problem where we assume the cost of rejection for any transaction is $f_2$.
Thus, rejecting three small transactions size $x$ costs thrice as much as rejecting a larger transaction of size $3x$. 
Finally, if there are still some funds left, they are added to $B_o^r$.


\RestyleAlgo{ruled}

\SetKwComment{Comment}{ }{ }
\SetKwInput{KwInput}{Input}
\SetKwInput{KwOutput}{Output}
\SetKwFunction{FDecide}{Decide}
\SetKwProg{Fn}{}{}{}

\begin{algorithm}[hbt!]
\caption{Decision on transaction}
  \label{Algorithm:decide_on_trx}
  \Fn{\Call{Decide}{$\tracker , x , B^{sdr}, B^{rcv}$}}{
 $Status \gets \mathtt{Accept}$\\
        \uIf{$ \frac{\tracker}{2^{i}} < x \leq \frac{\tracker}{2^{i-1}} \; and \; x \leq B_i^{sdr}$}{
         Accept $x$\\
        $ \nebu \gets \min(\tracker,B_i^{sdr}-x+B_o^{sdr})$ \\
        $B_o^{sdr} \gets \max(0,B_i^{sdr}-x+B_o^{sdr}-\tracker)$\\
        $B_i^{sdr} \gets \nebu$\\
}
    \uElseIf{$  x_i \leq \frac{\tracker}{C} \; and \; x \leq B_s^{sdr}$}{
        Accept $x$\\
        $\nebu \gets \min(2\tracker,B_s^{sdr}-x+B_o^{sdr})$\\ 
        $B_o^{sdr} \gets \max(0,B_s^{sdr}-x+B_o^{sdr}-2\tracker)$\\
        $B_s^{sdr} \gets \nebu$\\}
    \uElseIf{$  x_i \leq \frac{\tracker}{C} \; and \; x > B_s^{sdr}$}{
        Do off-chain rebalancing to fill $B_s$ and pay $f_2C$.\\
        $B_o^{rcv} \gets B_o^{rcv}-(2\tracker-B_s^{sdr}) $.\\ \label{decidealg:bo}
        $B_s^{rcv} \gets B_s^{rcv}-x $.\\ \label{decidealg:bs}
        Accept $x$\\
        $B_s^{sdr} \gets 2\tracker $.\\
        }
    \uElse{
    Reject x\\
    $Status \gets \mathtt{Reject}$}
    \KwRet $(B^{sdr}, B^{rcv}, Status)$
  }
\end{algorithm}

\RestyleAlgo{ruled}

\SetKwComment{Comment}{}{}
\SetKwInput{KwInput}{Input}
\SetKwInput{KwOutput}{Output}
\SetKwFunction{FHandle}{Handle\_funds}
\SetKwProg{Fn}{}{}{\KwRet}

\begin{algorithm}[hbt!]
\caption{Handling funds coming from the other side}

  \label{Algorithm:handle-fund}
  
  \Fn{\Call{HandleFunds}{$\tracker , x , B$}}{
        $\nebu \gets \min(2\tracker,B_s+x)$\\ \label{handlealg:bs}
        $x \gets \max(x+B_s-2\tracker,0)$\\
        $B_s \gets \nebu$\\
        \For{$i \in [\slog]$ in decreasing order}{
            \uIf{$x > 0$}
            {
              $\nebu \gets \min(\tracker,B_i+x)$\\
              $x \gets \max(x+B_i-\tracker,0)$\\
              $B_i \gets \nebu$\\
            }
        }
      $B_o \gets B_o + x$
    \KwRet $(B)$
  }
\end{algorithm}


\RestyleAlgo{ruled}

\SetKwFunction{FTracker}{Tracker}
\SetKwComment{Comment}{ }{ }
\SetKwInput{KwInput}{Input}
\SetKwInOut{KwInitialise}{Initialise}
\SetKwFunction{FDecide}{Decide}
\SetKwFunction{FTracker}{Tracker}
\SetKwFunction{FHandle}{Handle\_funds}

\begin{algorithm}[htb!]
\caption{Main algorithm }
  \label{Algorithm:logC-competitive-main}

\KwInitialise{left side buckets $B^{\ell}$}
\KwInitialise{right side buckets $B^{r}$}
\KwInitialise{tracker $\tracker, X \gets 0, \emptyset $}

\For{transaction $x$ in order of arrival}{
    concatenate $x$ to $X$\\
    $\tracker' \gets \Call{Funds}{X}$\\ \label{mainalg:oracle}
    \uIf{$\tracker' > \tracker$}{ 
        $\tracker \gets \tracker' + f_1$\\
        recharge to $2(2 + \slog)\tracker$ \label{mainalg:recharge}
    }
    $sdr,rcv \gets \ell , r$\\
    \uIf{$x$ is from right to left }{
     $sdr,rcv \gets r , \ell $
    }
        $ B^{sdr},B^{rcv}, Status \gets \Call{Decide}{\tracker ,x , B^{sdr},  B^{rcv}}$\\
        \uIf{$ Status == \mathtt{Accept} $}{
          $ B^{rcv} \gets \Call{HandleFunds}{\tracker , x , B^{rcv}}$
       }
    }
    
\end{algorithm}

\subsection{A $7 + 2\slog$ competitive online algorithm}
\begin{theorem}\label{thm:main}
  \Cref{Algorithm:logC-competitive-main} is $7 + 2\slog$ competitive.
\end{theorem}
\begin{proof}
  We know that for any $i$, $\costoff(X_i) \ge \costoff(X_{i-1})$.
  From \Cref{Lemma: Tracking} and \Cref{Lemma: Money cost}, we know that cost of \on for recharging (in \Cref{Algorithm:logC-competitive-main}) is at most $(5 + 2 \slog)\costoff(X_t)$.
  Let $t_1$ and $t_2$ (with $t_2>t_1$) be the two consecutive times \on recharges, then we show that the cost of \on for rebalancing and rejection is smaller than $2(\costoff(X_{t_2}) - \costoff(X_{t_1}))$.
  Then $\coston(X_t) \le (7 + 2\slog)\costoff(X_t)$.

  For every strategy of \off and any two consecutive recharging times $t_1$ and $t_2$, we show that the rebalancing and rejection cost of \on between times $t_1$ and $t_2$ is at most twice that of \off as defined by the strategy.
  Having the strategy of \off, we split the time between rechargings even further, into epochs; we will show that the competitive ratio of $2$ holds for every epoch.

  The left epoch starts with the first transaction that makes some bucket in $B^{\ell}$ non-full (smaller than the original amount); the left epoch ends either before \on recharges, or $B_o^{\ell} > 0$.
  In a left epoch, every transaction from the right side is accepted; non-fullness of some buckets on one side means $B_o > 0$ on the other side.
  The right epoch is defined similarly, but since the epochs are disjoint, we can prove the statement for a left epoch only.

  For transactions below $\frac{\tracker}{C}$, we argue that the cost of \on is at most the cost of \off.
  \on accepts everything, so \on pays only for rebalancing.
  \off either rebalances too, in which case the cost is the same as \on; or it rejected some transactions.
  Since \on starts with $2\tracker$ funds in $B_s$ and refills the bucket with the highest priority, this means \off rejected some transactions summing to at least $\tracker$.
  There are at least $C$ of them, so \off's cost is also above $Cf_2$.
  If there is a counterexample containing a small transaction that \off rejects, then we can modify it to a counterexample where the transaction is increased to $\frac{\tracker}{C}$.
  So we can show the ratio in the case that no small transactions are coming.

  Now that we have the strategy for \off (decisions before rebalancing), we define some variables that track the competitive ratio. We will look at incoming transactions, and prove that the competitive ratio is always below $2$.
  We say that a transaction $x$ belongs to a bucket $B_i$ if $\frac{\tracker}{2^i} < x \le \frac{\tracker}{2^{i-1}}$.
  A transaction is red if it is rejected by \on and accepted by \off and it is blue if it is accepted by \on and rejected by \off.
  Let $\rho_i$ ($\beta_i$) be the number of red (blue) transactions in the bucket $B_i$.
  We can disregard transactions for which \on and \off make the same decision.
  If the transaction is rejected by both, it improves the ratio.
  If the transaction is accepted by both, it can be simulated by decreasing $\tracker$.

  We prove by induction that $\sum_{k \le j} \rho_k \le 2 \sum_{k\le j} \beta_k$ for all $j$.
  We show that if the equality holds and \off has enough funds to accept incoming transactions, \on accepts too.

  Let us examine $j = \slog$.
  We know that the ratio between any two transaction sizes in $B_{\slog}$ is less than $2$, so any two red transactions are bigger than one blue.
  Moreover, all funds that arrived from the right side were put into $B_{\slog}$ (if it is not full).
  So if $\rho_{\slog} = 2 \beta_{\slog}$, \on has at least the amount of funds in $B_{\slog}$ \off has.
  To continue the induction for buckets with smaller indices, we reassign some red or blue transactions to different buckets.
  If $\rho_{\slog} < 2 \beta_{\slog}$, we move at most one red and some blue transactions (in decreasing order of size) to $B_{\slog -1}$, stopping just before $\rho_{\slog} \ge 2 \beta_{\slog}$.

  For general $j$, we know that, due to the reassignment, in every bucket smaller than $j$, \on rejected exactly twice the number of transactions \off did.
  Moreover, \off needs to use at least the same amount of funds to accept red transactions compared to funds needed by \on to accept blue ones.
  Now, in the bucket $B_j$ holds $\rho_j = 2\beta_j$.
  Again, we pair every two red transactions to one blue, such that the sum of red is bigger than blue.
  Before the reassignment, the ratio between any two transactions is at most $2$.
  The reassignment (if occurred) moved at most one red and at least one blue that is smaller than any original transaction in the bucket, so we can pair the moved red to moved blue.
  In transactions in buckets in $j$ and bigger, \off used more funds than \on.
  Any funds that arrived from the right side were put into some bucket in $j$ or below, so if \off has enough funds to accept, \on has too.

  The same argument holds for a right-epoch, and we note that epochs are disjoint and cover the entire transaction sequence between times $t_1$ and $t_2$. 
  Since we chose the consecutive recharging times $t_1$ and $t_2$ arbitrarily, the rebalancing and rejection cost of \on between any two consecutive rechargings is at most twice that of \off within the same period. 
  Therefore, \Cref{Algorithm:logC-competitive-main} is $7 + 2\slog$ competitive.
\end{proof}

Moreover, observe that depending on the value of $C$, we can tweak the algorithm.
There can be a smaller number of buckets that accept transactions in a bigger range.

\subsection{Lower bound}
In this subsection, we show how much funds \on needs to lock in the channel to have a chance to be $c$ competitive.
We make the construction for $A$, the amount of funds that \off locked in the channel.
Observe that \off would rather reject transactions that have average size larger than $\frac{A}{C}$ than perform off-chain rebalancing to accept them.

\begin{lemma}\label{Lemma: Lower bound two party}
  For any $A$, if \on's cost for rejection is at most $c$ times \off's cost for rejection (for $c < \frac{\log C}{\log \log C}$),
  any deterministic \on needs to lock at least $\sigma = A \cdot \left(\frac{\frac{1}{c+1}\log C}{\log c+1}+1\right)$ funds in the channel.
\end{lemma}
\begin{proof}
  We describe an epoch: \off starts with $A$ funds left, then some transactions are sent from left to right and finally one transaction of size $A$ is sent right to left.
  If the funds in the channel of \on is smaller than $\sigma$, then the cost of \on is more than $c$ times that of \off.

  One epoch consists of at most $\frac{\log C}{\log c+1}+1$ phases.
  In phase $i$ (starting from $i=0$), there are $(c+1)^i$ transactions of size $\frac{A}{(c+1)^i}$.
  \off always accepts all transactions in the latest phase.
  If at the end of any phase, the cost of \on is more than $c$ times of \off, then a transaction of size $A$ is sent back and another epoch starts.
  Observe that for $c < \frac{\log C}{\log \log C}$, after rebalancing the epoch ends too (because the cost is $c$ times bigger).
  We can assume this cannot happen, \on does not perform off-chain rebalancing.

  We compute how much funds \on needs to stay within the competitive ratio until the last phase (where transactions of sizes $\frac{A}{C}$ are sent).
  After phase $i$, \off accepted $(c+1)^i$ transactions and rejected $\frac{(c+1)^i - 1}{c}$,
  so \off can reject up to $(c+1)^i - 1$ transactions among $\frac{(c+1)^{i+1}-1}{c}$.
  So \on has to accept at least $\frac{(c+1)^{i+1} - 1 - c(c+1)^{i} + c}{c} = (c+1) \frac{(c+1)^{i-1} - 1}{c}$ transactions.
  
  The size of transactions is decreasing, so optimally, \on accepts transactions when they are needed.
  So in phase $i+1$ it needs to accept $(c+1)\frac{(c+1)^i - 1}{c} - (c+1) \frac{(c+1)^{i-1} - 1}{c} = (c+1)^i$ transactions.
  Of course, it needs to accept the transaction in the phase $0$.

  The cost of transactions accepted by \on in phase $i>0$ is $(c+1)^{i-1} \frac{A}{(c+1)^i} = \frac{A}{c+1}$.
  To maintain the competitive ratio \on needs to accept transactions worth $A + A\frac{\log C}{\log c+1}$ in total.
  Independently of the cost of \off and \on before, there can be
  one epoch after another where the ratio is worse than $c$,
  so at the end, the ratio would be above $c$.
\end{proof}

\begin{theorem}
  There is no deterministic $c$-competitive algorithm for $c \in o(\sqrt{\log C})$.
\end{theorem}
\begin{proof}
  From \Cref{Lemma: Lower bound two party} for any $A$, \on needs $A \cdot (\frac{\frac{1}{c+1}\log C}{\log c+1}+1)$ funds
  to have its rejection cost $c$-competitive.
  But \on also needs to lock some funds in the channel.
  The total cost is then $c + A (\frac{\frac{1}{c+1}\log C}{\log c+1}+1)$, which is bigger than $\calO(\sqrt{\log C})$.
\end{proof}

\section{Evaluation}\label{sec:eval}

We now complement our theoretical worst-case analysis and study the performance of our algorithms empirically. 

\subsection{Methodology}

In order to study average case behavior, we consider the performance of our main algorithm \on (\Cref{Algorithm:logC-competitive-main}) on randomly generated transaction sequences in the general bidirectional transaction stream setting.
Our benchmark is the performance of the optimal offline algorithm \off, however, since computing optimal solutions is NP-hard~\cite{oracle}, computing the achieved competitive ratio in simulations is computationally expensive. 
In order to have a baseline nevertheless, we use dynamic programming to exactly and efficiently compute the 
cost of \off for sequences of reasonable lengths. 
For ease of presentation, we defer this algorithm to \Cref{app:dp}.
We additionally conduct experiments using actual Lightning Network data, and considering realistic parameters.

\subsection{Comparison of \on to \off}

\paragraph{Average performance of \on}
We first report on the competitive ratio our online algorithm achieves under random transaction sequences. 
We begin by sampling $50$ random transaction sequences of length $50$ each. 
In each sequence, transaction sizes are first sampled independently from the folded normal distribution with mean $0$ and standard deviation $3$, and then we sample the direction of the transaction (left-to-right or right-to-left) uniformly at random. 
Finally, we quantise the size of the transaction to the closest integer. 
We then run both \off and \on on the generated sequences and compute the average of the following:
\begin{enumerate}
    \item cost (sum of rejection/off-chain rebalancing/on-chain recharging costs)
    \item sum of funds locked in the channel ($A(X)$)
    \item acceptance rate (fraction of transactions accepted)
    \item funds that are transferred across the channel by off-chain rebalancing per sequence
    \item number of on-chain rechargings per sequence 
\end{enumerate}

We present our results in \Cref{table:off-on-comparision}. 
As we can see from the cost of \on vs \off in \Cref{table:off-on-comparision}, the competitive ratio is generally significantly lower than the  $7 + 2\slog$ bound as suggested by our conservative worst-case analysis in \Cref{thm:main}. 

\paragraph{Heuristics to improve the performance of \on}
However, we notice in our experiments that \on seems to overcharge the channel.
This is most noticeable when we observe the effect of $C$ on the performance of \on. 
From \Cref{table:off-on-comparision}, increasing $C$ in a range of medium (not too small) values does not change \off's cost noticeably. 
In contrast, both the average cost and total amount of locked funds of \on grows with $C$. 
This is due to the fact that \on uses $(4 +2 \slog, f_1)$-recharging to ensure that it always has significantly more funds than \off, even though a big fraction of these funds remain unspent.
\on is also limited by the fact that it does not borrow funds from other buckets when a bucket is depleted.
For instance, \on always charges $B_s$ to $2\tracker$ and only uses these funds to accept transactions that are smaller than $\frac{\tracker}{C}$.
Thus, as $C$ increases, the number of transactions that fall into the $B_s$ bucket decreases and the funds in $B_s$ remain unspent.

These observations motivate us to design a less pessimistic version of \on that we expect will perform better than \on. 
We introduce \oni which is a slightly altered version of \on: \oni follows the $(\slog, f_1)$-recharging algorithm and does not divide the funds into separate buckets. 
Instead, \oni accepts all the transactions smaller than $\tracker$ as long as it has funds.
Otherwise, if the transaction is small ($< \frac{\tracker}{C}$) it off-chain rebalances to fill the bucket and accepts the transaction.
Similar to \on, \oni rejects a transaction if it is larger than $\frac{\tracker}{C}$ . 

Our empirical results in \Cref{table:off-on-comparision} confirms that the average cost of \oni is significantly smaller than \on.
The acceptance rate of \oni is slightly smaller than \on, which is expected as \oni does not have separate funds for each range of transactions (as defined by the buckets), and as a result might miss some transactions. 
We observe that \oni performs more off-chain rebalancings compared to \on on average because \oni does not reserve separate funds for small transactions. 
However, one issue with both \on and \oni is that they recharge the channel too often (as soon as $\tracker < A(X_t)$). 
As can be seen from \Cref{table:off-on-comparision} ($f_2=2$), both algorithms perform more than $5$ rechargings on average for transaction sequences of length $50$. 
This increases the cost of both algorithms significantly as each recharging instance incurs a cost of at least $f_1$. 

We thus design another version of \oni to address the aforementioned problem. 
\onii works exactly as \oni except that it does not recharge the channel as frequently as \oni does. 
\onii only recharges the channel if  $\alpha \cdot \tracker < A(X)$, where $\alpha > 1$ is some constant that controls the how often the algorithm recharges the channel and can be fine-tuned empirically based on $f_1$ and $f_2$. 
If we set $\alpha = 1$, \onii becomes equivalent to \oni and has higher acceptance rate.
This is favorable when $f_2$ is large and $f_1$ is small. 
Conversely, by increasing $\alpha$, \onii recharges the channel less frequently but the acceptance rate falls.
This is favorable when $f_1$ is large and $f_2$ is small. 
In our experiments, we observe that for the case $f_2=2,C=2$, when all the other parameters are as \Cref{table:off-on-comparision}, $\alpha = 2$ yields the lowest average cost.
Thus, in our evaluation of \onii, we use $\alpha=2$ and from \Cref{table:off-on-comparision} we note that this choice of $\alpha$ halves the number of rechargings compared to \oni, which consequently leads to lower average cost.
Additionally, we note that the total amount of funds in the channel of \onii is close to \off. 

\paragraph{Varying the distribution of the generated sequences}
We also evaluate how the performance of our algorithms is affected by the variance of the transaction size.
We sample $50$ sequences each of length $50$ with each transaction in the sequence independently sampled from the folded normal distribution with mean 0 and standard deviation $\sigma$, for a range of $\sigma$ values across $[3,20]$.
The transactions are quantised to the closest integer, and the direction of each transaction is sampled uniformly at random.
We then observe the cost of \on, \oni,\onii and \off. 
As can be seen from in \Cref{fig:costVSstd}, the cost of all $4$ algorithms rises as $\sigma$ increases. 
This is due to the fact that increasing the variance of the sampled transactions reduces the probability of getting a similarly sized transaction coming from the other side, thus increasing the speed at which the balance on one side gets depleted.
We note, however, that \Cref{fig:costVSstd} shows that even for large values of $\sigma$, \on's average cost remains a lot smaller than the worst case upper bound of $7 + 2\slog$.
We also observe that the cost of \oni and \onii is noticeably smaller than \on and grows at a much slower pace with respect to increases in $\sigma$.

Another factor we look at is how the asymmetry of the transaction flow along a channel can affect the performance of our algorithms.
To do so, we generate $50$ sequences each of length $50$, and sample the size of each transaction from a folded normal distribution with mean $0$ and standard deviation $3$, and then quantise the transaction to the nearest integer. 
We then sample the direction of the transactions according to a Bernoulli distribution with parameter $p$, where $p$ represents the probability of sampling a left-to-right transaction. 
We see from \Cref{fig:costVSp} that the cost of all algorithms decrease as $p$ increases from $0$ to $0.5$.
As $p$ increases from $0.5$ to $1$, the cost function increases again. 
This conforms to our intuition that extremely asymmetric sequences are harder to handle as the lack of sufficiently many transactions from one side just increases the speed at which the balance on the other side gets depleted.
We observe that both \oni and \onii nevertheless perform comparatively better than \on when given these asymmetric transaction sequences (see \Cref{fig:costVSp}).

\begin{table*}[htb!]
\centering

\begin{tabular}{||p{0.3cm} |p{0.2cm}|| p{1.2cm}| p{0.9cm}| p{1.1cm}| p{1.6cm}| p{1.9cm}|| p{1.2cm}| p{0.9cm}| p{1.1cm}| p{1.6cm} | p{1.9cm} ||} 
 \hline
  \multicolumn{2}{||c||}{Param} & \multicolumn{5}{|c||}{\off} & \multicolumn{5}{|c||}{\on} \\
 \hline\hline
 C & $f_2$ & Cost & $A(X)$ & Accept rate&  Off-chain rebalancing &  Rechargings &  Cost & $A(X)$ & Accept rate & Off-chain rebalancing &  Rechargings  \\ [0.5ex] 
 \hline
 $2$  & $0.5$  & $\textbf{15.02}$ & $6.4$ & $0.78$ & $0.8$ & $1$ & $\textbf{63.3}$ & $44.26$ & $0.50$ & $0.9$ & $2.18$ \\ 
 \hline
 $8$  & $0.5$ & $\textbf{15.21}$  & $6.38$ & $0.77$ & $0$ & $1$ & $\textbf{87.79}$  & $69.06$ & $0.50$ & $0$ & $2.04$\\
 \hline
 $2$  & $2$ & $\textbf{23.6}$  & $14.26$ & $0.95$ & $5.36$ & $1$ & $\textbf{127.02}$  & $100.2$ & $0.91$ & $0.38$ & $5.86$ \\ 
 \hline
 $8$  & $2$ & $\textbf{24.5}$  & $13.9$ & $0.92$ & $0$ & $1$&  $\textbf{184.32}$  & $156.6$ & $0.9$ & $0$ & $5.84$ \\ 
 \hline
\end{tabular}
\\
\begin{tabular}{||p{0.3cm} |p{0.2cm}|| p{1.2cm}| p{0.9cm}| p{1.1cm}| p{1.6cm}| p{1.9cm}|| p{1.2cm}| p{0.9cm}| p{1.1cm}| p{1.6cm} | p{1.9cm} ||}  
 \hline
  \multicolumn{2}{||c||}{Param} & \multicolumn{5}{|c||}{\oni} & \multicolumn{5}{|c||}{\onii} \\
 \hline\hline
 C & $f_2$ & Cost & $A(X)$ & Accept rate& Off-chain rebalancing &  Rechargings & Cost & $A(X)$ & Accept rate &  Off-chain rebalancing  &  Rechargings  \\ [0.5ex] 
 \hline
 $2$  & $0.5$  & $\textbf{30.74}$ & $6.86$ & $0.44$ & $11.18$  & $2.18$ & $\textbf{26.52}$ & $5.56$ & $0.41$ & $10.22$ & $1.2$ \\ 
 \hline
 $8$  & $0.5$ & $\textbf{39.98}$  & $19.86$ & $0.48$ & $0$  & $2.04$ & $\textbf{32.94}$  & $15.9$ & $0.46$ & $0$ & $1.18$ \\
 \hline
 $2$  & $2$ & $\textbf{66.16}$  & $14.42$ & $0.84$ & $25.48$ & $5.86$ & $\textbf{60.38}$  & $11.3$ & $0.78$ & $27.16$ & $2.2$ \\ 
 \hline
 $8$  & $2$ & $\textbf{81.35}$  & $42.72$ & $0.89$ & $1.5$ & $5.84$ &  $\textbf{58.38}$  & $33.3$ & $0.83$ & $1.56$ & $2.16$ \\

 \hline
\end{tabular}
\caption{Comparison between the performance of \off, \on, \oni and \onii on randomly generated transaction streams. 
The result is averaged over $50$ sequences each of length $50$. 
The size of each transaction is independently sampled from the folded normal distribution with mean $0$ and standard deviation $3$, then quantised to the closest integer. 
We set $f_1=3$ and $R=0$. 
$A(X)$ is the total amount of funds in the channel from recharging the channel. 
``Accept rate" shows the average fraction of transactions that were accepted. 
``Off-chain rebalancing" shows how much funds on average was moved along the channel using off-chain rebalancing. 
``Rechargings" shows the average number of rechargings performed. 
Note that since \off knows the entire sequence in advance, it only recharges the channel once at the beginning of each sequence.
}
\label{table:off-on-comparision}
\end{table*}

 \begin{figure}
\centering
  \includegraphics[scale=0.7]{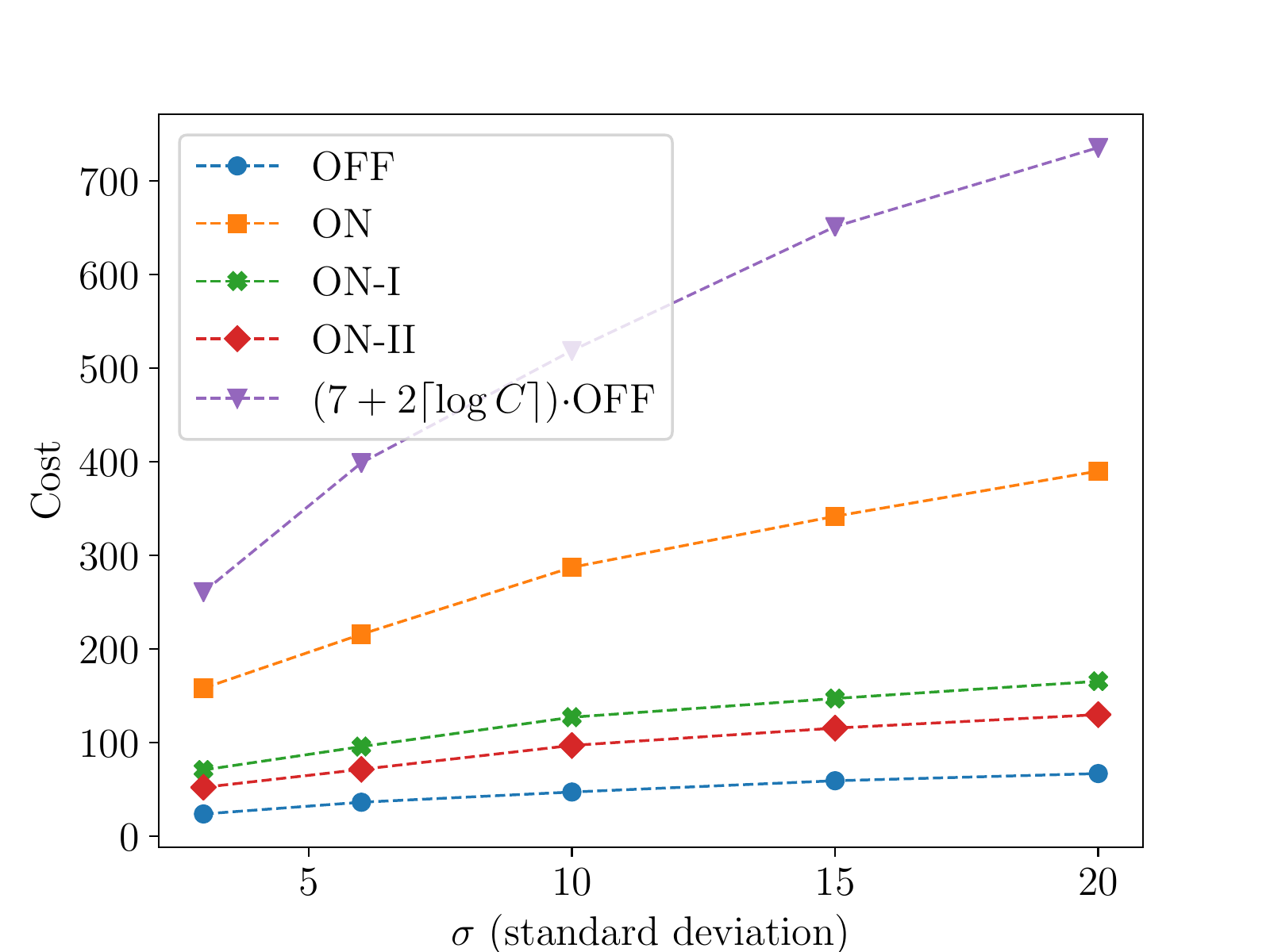}
  \caption{Average cost of our algorithms over $50$ randomly generated transaction streams, each of length $50$. The size of each transaction is independently sampled from the folded normal distribution with mean $0$ and standard deviation $\sigma$. We use the parameters $f_1 = 3 , f_2 = 2 , R = 0 , C = 4 , \alpha=2$.}
  \label{fig:costVSstd}
 \end{figure}

\begin{figure}
\centering
  \includegraphics[scale=0.7]{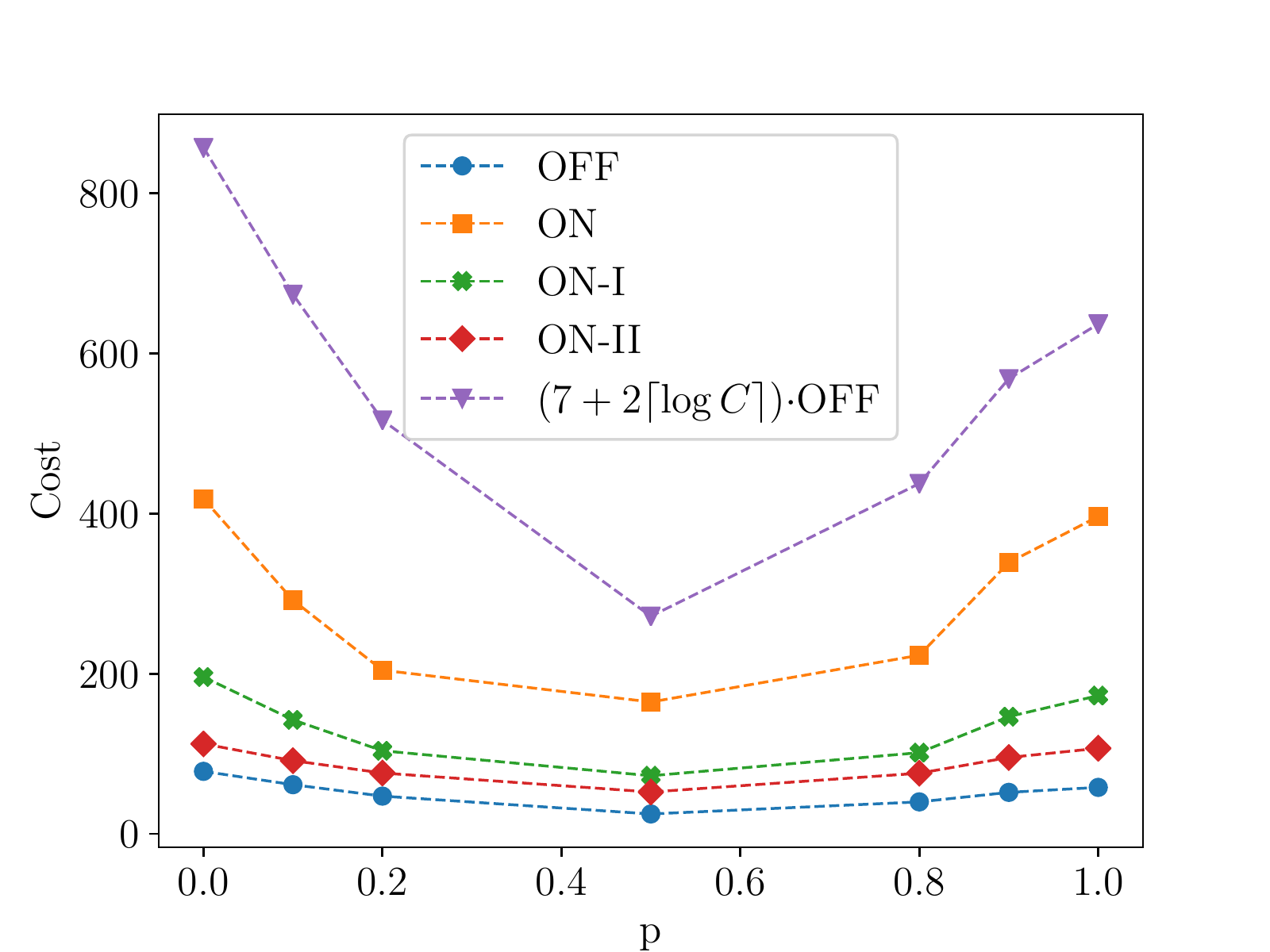}
  \caption{Average cost of our algorithms over $50$ randomly generated transaction streams each of length $50$. The size of each transaction is sampled from the folded normal distribution with mean $0$ and standard deviation $3$. $p$ is the probability of sampling a left-to-right transaction. We use the parameters $f_1 = 3 , f_2 = 2 , R = 0 , C = 4 , \alpha = 2$.}
  \label{fig:costVSp}
 \end{figure}

\subsection{Case study: Lightning Network}\label{sec:cycle_length_estimation}
We conclude our evaluation section with a case study of the Lightning network. 
We first run our experiments with realistic parameters taken from Lightning Network data.
In the Lightning Network, $f_1$ is the on-chain transaction fee (roughly around $1000$ satoshi) which is a lot larger than $f_2$, the base fee one receives when forwarding a payment (around $1$ satoshi).
Thus, when we evaluate our algorithms using similar parameters for short transaction sequences, we observe that \off rejects all transactions and, in fact, will not even open a channel as it would incur a larger cost as compared to simply rejecting transactions. 

Next, we conduct an empirical estimation of the average rebalancing cycle length in the Lightning Network using the latest snapshot (September 2021) of the Lightning Network from the Lightning Network gossip repository~\cite{lngossip}.
This is important as $C$ is an integral component of the competitive ratio in our main theorem, \Cref{thm:main}.
The Lightning Network snapshot contains $117,894$ channels, including pairs of users with more than one channel between them. 
After merging these channel pairs, we get $63,820$ channels. 
We compute the length of the shortest cycle containing both $\ell$ and $r$ for every pair of channel holders $(\ell,r)$.
\Cref{table:cycle} shows the frequency of each cycle length (including the channels that are not part of any cycle at all).
From \Cref{table:cycle}, we see that a large fraction of channel holders are part of cycles of length at least $4$. 
We note, however, that actual cycles might be longer as precise balance information is hidden in the Lightning Network snapshot. 

\begin{table}[htb!]
\centering
  \begin{tabular}{||c||c|c|c|c|c||} 
 \hline \hline
 \textbf{Cycle length} & $4$  & $5$ & $6$ & $7$ & N.A. \\
 \hline
 \textbf{Frequency} &  $49,424$ & $7,758$ & $469$ & $12$ & $6,157$\\
 \hline \hline
\end{tabular}

\caption{Frequency of the length of shortest cycle between all users in the Lightning Network. The last column shows the frequency of channels that are not part of any cycle. The average cycle length is 4.15}
\label{table:cycle}
\end{table}

\section{Conclusion}\label{sec:conclusion}
This paper presents competitive strategies to maintain minimise cost while maximising liquidity and transaction throughput in a payment channel.
Our algorithms come with formal worst-case guarantees, and also perform well in realistic scenarios in simulations. 

We believe that our work opens several interesting avenues for future research.
On the theoretical front, it would be interesting to close the gap in the achievable competitive ratio, and to explore the implications of our approach on other classic online admission control problems.
Furthermore, while in our work we have focused on deterministic algorithms, it would be interesting to study the power of randomised approaches in this context, or to consider different adversarial models.

\newpage
\bibliographystyle{acm}

\nocite{*} 

\balance
\newpage
\appendix

\section{Optimality of deterministic algorithms in a unidirectional stream with rejection}\label{app:rejection_sketch}

Sketch proof of Conjecture \ref{conjecture:rejection}.

\begin{proof}[Sketch]
  The best \on algorithm needs to recharge the channel when \off does.
  If it recharges the channel later, it incurs cost that \off is not incurring, so the competitive ratio worsens.
  If it recharges sooner, there exists a sequence that either forces \off to waste funds or incur a big cost.
  
  After \off recharges, it can reconsider and accept previously rejected transactions,
  but \on needs to reject them.
  Now, the situation is similar as in the case without rejection.
  \on needs to recharge, but it already paid for some rejections whereas \off pays only for recharging and accepting very small transactions.
  
  \on disadvantaged in this way cannot achieve a better competitive ratio than $2 + \frac{\sqrt{5}-1}{2}$
\end{proof}

\section{Computing the cost of \off using dynamic programming}\label{app:dp}
In this section, we describe a dynamic programming algorithm \on that solves the main problem.
We assume that the size of transactions is integer (moreover the sum of transactions should be small).

Let $\textsc{Cost}_{i}(F_{\ell} , F_r)$ be the minimum cost for rejecting and off-chain rebalancing in processing sequence $X_i$ that ends with $b(\ell) = F_{\ell}$ and $b(r) = F_r$
(For values of $F_{\ell}$ and $F_r$ smaller that $0$, we define it to be $\infty$).
$\textsc{Cost}_{i}(F_{\ell} , F_r)$ can be derived from $\textsc{Cost}_{i-1}$ given the decision on the $i$'th transaction.

Let us assume wlog that $i$'th transaction is from $\ell$ to $r$.
\off has three choices when encountering $x_i$.
The first option is to reject $x_i$, the the cost is $A_1 = \textsc{Cost}_{i-1}(F_{\ell} , F_r) + Rx_i+f_2$.
The second option is to accept $x_i$ which gives cost $A_2 = \textsc{Cost}_{i-1}(F_{\ell}+x_i , F_r-x_i)$.
The last option is to off-chain rebalance before $x_i$ and then accept $x_i$.
Note that any off-chain rebalancing before rejecting or accepting (while having enough funds) can be postponed.
This gives cost $A_3=\min_{a \le F_{\ell} + x_i}  \textsc{Cost}_{i-1}(a, F_r+ F_{\ell} - a) + C\cdot R(F_{\ell} + x_i - a)+Cf_2$.
\off chooses the best option, that means $\textsc{Cost}_{i}(F_{\ell}, F_r)=\min{\{A_1,A_2,A_3\}}$

We handle right to left transaction in the same way.

Given the previous, \on computes $\textsc{Cost}_{t}$ for all valid pairs $(F_{\ell} , F_r)$ and the final cost is
\[
  \costoff(X_i)  = \min_{F_{\ell},F_r}  \textsc{Cost}_{t}(F_{\ell},F_{r}) + F_{\ell} + F_{r} + f_1
\]

To bound the time complexity of \on, we observe some bounds for $S = F_{\ell}^* + F_r^*$, where $F_{\ell}^*$ and $F_r^*$ are the values of $F_{\ell}$ and $F_r$ achieving the minimal cost.
Observe that $S \le \sum_{i\le t} x_i$, it is not worth to have more money than the sum of the trasactions.
We can strenghten the inequality and instead of $\sum_{i\le t}x_i$, we can compute the minimal amount needed to accept every transaction.
The other option is to reject everything, so we know that $f_1 + S \le \sum_{i\le t} Rx_i + f_2$.

Now we can prove the theorem about the described algorithm \on.
\begin{theorem}
  \on computes the optimal cost in time $\calO(tS^3)$, where $S$ is the bound on the maximal funds in the channel and $t$ the number of transactions.
\end{theorem}
\begin{proof}
  In the dynamic programming, we take into account all possible decisions \on can make, by this, correctness follows.

  The algorithm \on tries all possible amounts between $1$ and $S$ and starting distributions.
  There are $S^2$ of them.
  While computing one value, it needs to look at at most $S$ precomputed values.
  And it needs to do it at most $t$ times.
\end{proof}

Using dynamic programming for calculating \off has two advantages.
First, we can easily recover the decisions of \off.
Secondly, dynamic programming provides us with optimum solution for all subsequences of $X_t$. This is useful for implementing \Cref{Algorithm:logC-competitive-main}.

\end{document}